\documentclass[journal]{IEEEtran}

\usepackage{amsmath,amsfonts,amssymb,amsthm}
\usepackage[pdftex]{graphicx}
\usepackage[dvipsnames]{xcolor}
\usepackage{cite}
\usepackage[shortlabels]{enumitem}
\usepackage{subcaption}
\usepackage{booktabs}
\usepackage[ruled,vlined,linesnumbered]{algorithm2e}
\usepackage{empheq}
\usepackage{titlesec}

\graphicspath{{./Figures/}}			

\newtheorem{theorem}{Theorem}

\newtheorem{remark}{Remark}

\newcommand{\mc}[1]{\mathcal{#1}}
\newcommand{\mbb}[1]{\mathbb{#1}}
\newcommand{\mrm}[1]{\mathrm{#1}}
\newcommand{\RHinf}{\mc{RH}_{\infty}}

\newcommand{\bbC}{\mbb{C}}

\newcommand{\Set}[2]{\ensuremath{\left\{ \vphantom{#2} #1 \right. \, \left| \, \vphantom{#1} #2 \right\}}}

\let\OLDthebibliography\thebibliography
\renewcommand\thebibliography[1]{
  \OLDthebibliography{#1}
  \setlength{\parskip}{0pt}
  \setlength{\itemsep}{0pt plus 0.1ex}
}
\makeatletter
\g@addto@macro\normalsize{%
  \setlength\abovedisplayskip{5pt plus 1pt minus 1pt}
  \setlength\belowdisplayskip{\abovedisplayskip}
  \setlength\abovedisplayshortskip{0pt plus 1pt}
  \setlength\belowdisplayshortskip{\abovedisplayskip}
}
\titlespacing{\section}{0pt}{0.5ex}{0.5ex}
\titlespacing{\subsection}{0pt}{0.5ex}{0.5ex}
\titlespacing{\subsubsection}{0pt}{1ex}{1ex}
\setlist[itemize]{nosep}
\setlist[enumerate]{nosep}

\title{Frequency Response Data Based LPV Controller Synthesis Applied to a Control Moment Gyroscope}

\begin{document}

%
%
%

\author{Tom~Bloemers,~Roland~T\'oth~and~Tom~Oomen
\thanks{T. Bloemers (corresponding author) and R. T\'oth are with the Control Systems Group, Department of Electrical Engineering, Eindhoven University of Technology. T. Oomen is with Control Systems Technology, Department of Mechanical Engineering, Eindhoven University of Technology, email: \texttt{\{t.a.h.bloemers, r.toth, t.a.e.oomen\}@tue.nl}. R. T\'oth is also with the Systems and Control Laboratory, Institute for Computer Science and Control, Kende u. 13-17, H-1111 Budapest, Hungary.
}
\thanks{This work has received funding from the European Research Council
(ERC) under the European Union's Horizon 2020 research and innovation
programme (grant agreement nr. 714663).}
}
\maketitle

\begin{abstract}
	Control of systems with operating condition-dependent dynamics, including control moment gyroscopes, often requires operating condition-dependent controllers to achieve high control performance. The aim of this paper is to develop a frequency response data-driven linear parameter-varying control design approach for single-input single-output systems, which allows improved performance for a control moment gyroscope. A stability theory using closed-loop frequency response function data is developed, which is subsequently used in a synthesis procedure that guarantees local stability and performance. Experimental results on a control moment gyroscope demonstrate the performance improvements. 
\end{abstract}


%
\IEEEpeerreviewmaketitle

\vspace{-1.5mm}
\section{Introduction}
%
%
%
%
Control of systems with operating condition-dependent dynamics, including control moment gyroscopes (CMGs), often requires operating condition-dependent controllers to achieve high control performance. CMGs are attitude control devices used, e.g., to control the attitude of spacecraft \cite{kristiansen2005comparative}. A CMG, see Figure \ref{fig:gyroscope}, consists of a rotating disk, which, when spinning, generates an angular momentum. The disk is mounted in a gimbal assembly which can rotate around multiple axes. Changing the direction of the angular momentum vector, through actuation of the gimbals, generates a gyroscopic torque \cite{wie2008space}. This torque can be used to, e.g., change the attitude of a spacecraft. The associated dynamics are nonlinear and characterized by coupled behavior and challenging rotational dynamics that change based on the operating conditions of the system. Locally, these behaviors manifest in terms of operating condition-dependent resonant dynamics, also commonly encountered in mechatronic systems \cite{toth2011lpv}. Flexible phenomena introduce severe practical limitations on the achievable performance, which become even more severe in case of operating condition-dependent dynamics. 
Achieving stability and high performance for these systems requires operating-condition dependent controllers \cite{abbas2014lpv,preda2017robust}.

The paradigm of linear parameter-varying (LPV) systems has been established to provide a systematic framework to efficiently handle operating condition-dependent nonlinear dynamics. LPV systems are characterized by a linear input-output (IO) map, whose dynamics depend on an exogenous time-varying signal. This scheduling variable $p$ can be used to capture the nonlinear or operating condition-dependent dynamics of a system. Typically, a priori information on the scheduling variable is known, such as the range of variation. LPV systems are supported by a well-developed model-based control and identification framework, with many successful applications, see \cite{hoffmann2015survey,toth2011lpv}. Model-based control techniques require an accurate parameter-dependent parametric model of the system suitable for LPV control design. In fact, obtaining such a high accuracy model is a challenging task, even for linear time-invariant (LTI) systems \cite{oomen2018advanced}.

Frequency response function (FRF) measurements enable systematic design of controllers directly from measurement data and are commonly employed in the industry \cite{oomen2018advanced}. A frequency response function estimate provides an accurate nonparametric description of the system that is relatively fast and inexpensive to obtain \cite{System_Identification_Frequency_Domain}. Also the nonparametric identification of local FRF measurements for LPV systems has been investigated in \cite{van2017accurate}, assuming that the underlying behavior is a smooth function of the scheduling variable. For the CMG, FRFs of the local dynamics can be accurately captured at a set of operating points. FRFs enable the use of classical techniques such as loop-shaping, alongside graphical tools including the Bode diagram or Nyquist plot, to design controllers \cite{maciejowski1989multivariable}. These controllers often have a proportional-integral-derivative (PID) structure in addition to higher-order filters to compensate parasitic dynamics. These methods have in common that the design procedure can be difficult as they are based on design rules, insight and experience.

Data-driven control design based on FRF measurements provides systematic approaches to design and synthesize LTI controllers. From a modeling perspective, data-driven control synthesis provides an alternative to control-oriented identification \cite{hjalmarsson2005experiment}. At first, the development of these methods have been along the lines of the classical control theory to tune PID controllers \cite{grassi2001integrated}. Later, these methods have been tailored towards more general control structures that focus on $\mathcal{H}_\infty$-performance \cite{khadraoui2014model}. The incorporation of model uncertainties into the control design enables the synthesis of stabilizing controllers that achieve sufficient robustness to account for the variations in the plant \cite{karimi2007robust,karimi2018robust}. Robust control methods are attractive to accommodate the operating condition-dependent resonant behaviors encountered in CMGs. A major drawback is a tradeoff between robustness and performance.

Including operating condition-dependent behavior in the data-driven control design framework is promising to overcome the tradeoff between robustness and performance. In \cite{formentin2016direct}, a time-domain approach is employed to identify an LPV controller such that the closed-loop mimics an ideal behavior. In \cite{kunze2007gain,karimi2013hinf,bloemers2019towards_lpv_synthesis}, frequency-domain control synthesis approaches are investigated. Common drawbacks are their limitations to stable systems only, conservative stability and performance constraints and the controller parameterization only allows for shaping of the zeros and not the poles.

Although frequency-domain data-driven controller synthesis enables powerful and systematic design approaches in the LTI framework, methods within the LPV framework are limited and conservative. The aim in this paper is to develop a data-driven LPV control design method that allows both for stable and unstable systems, applicable to an experimental CMG setup. Key steps are (i) a global LPV controller parameterization, which allows tuning of both the zeros and poles based on local information, and (ii) developing necessary and sufficient stability and performance analysis conditions.

The main contributions of this paper are
\begin{enumerate}[{C1)},left=\parindent]
	\item {\label{Contribution:1}} A procedure to synthesize LPV controllers for (possibly) unstable single-input single-output (SISO) plants from frequency-domain measurement data, with local internal stability and $\mc{H}_\infty$-performance guarantees. 
	\item {\label{Contribution:2}} Highlighting the advantages of using an LPV controller through application to an experimental CMG setup.
\end{enumerate}
This is achieved by the following sub-contributions. 
\begin{enumerate}[{C1)},left=\parindent]
\setcounter{enumi}{2}
	\item {\label{Contribution:3}} Developing of a local LPV frequency-domain stability analysis condition.
	\item {\label{Contribution:4}} Development of a local LPV frequency-domain $\mc{H}_\infty$-performance analysis condition.
\end{enumerate}
Contributions \ref{Contribution:3} and \ref{Contribution:4} are generalizations to the results presented in \cite{karimi2018robust, rantzer1994}. Specifically, when both the plant and controller are LTI the results in \cite{karimi2018robust} are recovered, and the results in \cite{rantzer1994} are recovered as a special case for stable systems. Other important differences in this paper are new insights and proofs of these sub-contributions, which establish links to the robust control theory and the B\'ezout identity. A global LPV controller parameterization in combination with \ref{Contribution:3} and \ref{Contribution:4} constitutes to \ref{Contribution:1}. Application of the developed procedures on an experimental CMG constitutes to \ref{Contribution:2}.


\textbf{Notation:} Let $\mbb{R}$ denote the set of real numbers and $\mbb{C}$ the set of complex numbers. Let $\bbC_0$ denote the imaginary axis and $\bbC_+$ the open right half-plane. The real part of a complex number $z \in \bbC$ is denoted by $\Re \{ z \}$. The set of proper, stable and real-rational transfer functions is denoted by $\RHinf$. 

\begin{remark}
	Although the theory in this paper is presented in continuous-time, a discrete-time equivalent is conceptually straightforward. Simply replace the variables $s$ with $z$, $i\omega$ with $e^{i\omega}$ and evaluate the frequencies along the unit circle instead of the imaginary axis, i.e., for the set $\Omega := \Set{\omega}{0 \leq \omega < 2\pi}$.
\end{remark}

\section{Problem formulation}
\label{section:ProblemFormulation}

\subsection{Control Moment Gyroscope}
\begin{figure}
	\begin{subfigure}[b]{.39\columnwidth}
		\centering
		\includegraphics[scale=0.06]{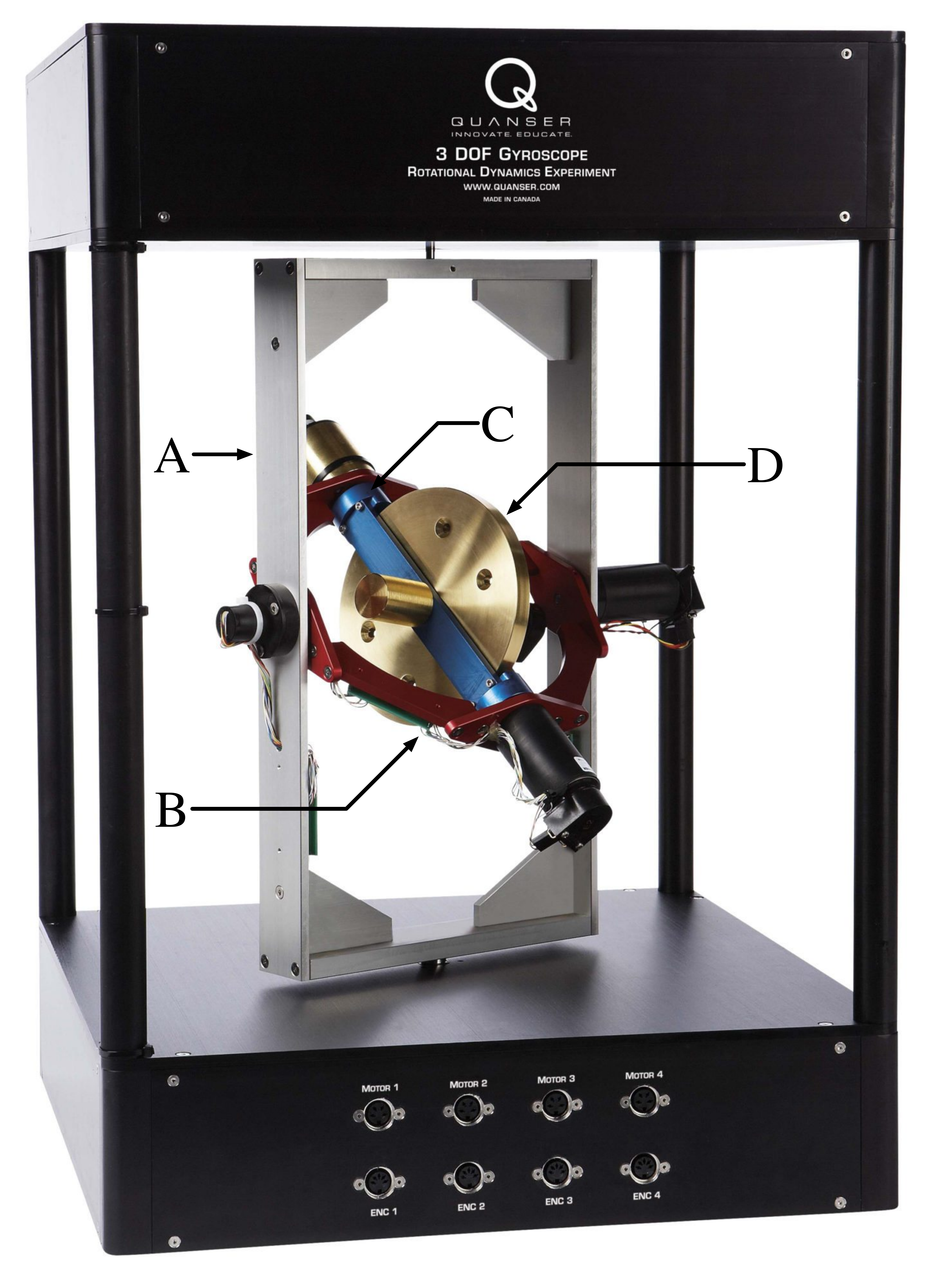}
		\caption{}
		\label{fig:gyroscope}
	\end{subfigure}
	\begin{subfigure}[b]{.60\columnwidth}
		\centering
		\includegraphics[scale=1.00]{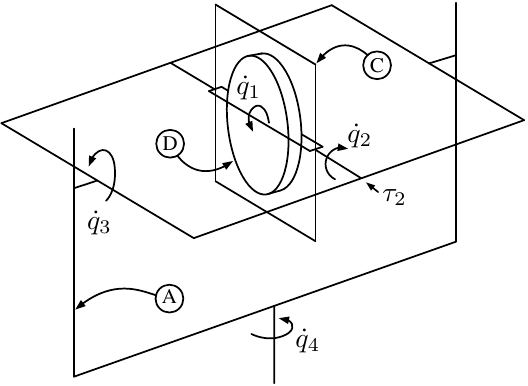}
		\caption{}	
		\label{fig:gyroscopeSchematic}
	\end{subfigure}
	\caption{Laboratory scale control moment gyroscope by Quanser (a); and schematic overview of the 3-DOF gyroscope (b).}
	\vspace{-3ex}
\end{figure}%
Figure \ref{fig:gyroscope} depicts the considered 3 degree of freedom (DOF) control moment gyroscope. It is comprised of a disk, $D$, which is mounted in a gimbal assembly consisting of three gimbals $C, B$ and $A$, corresponding to the schematic overview in Figure \ref{fig:gyroscopeSchematic}. The disk $D$ rotates with velocity $\dot{q}_1$, generating an angular momentum proportional to $\dot{q}_1$. Angle $q_2$ of gimbal $C$ is controlled through input torque $\tau_2$. Gimbal $B$ is assumed to be fixed in place such that $q_3 \equiv 0$, as depicted in Figure \ref{fig:gyroscopeSchematic}. Angle $q_4$ of gimbal $A$ is controlled through a gyroscopic torque, generated by changing angle $q_2$. As the disk tilts, a change in angular momentum causes gyroscopic torque, which is used to position gimbal $A$.

The equations of motions are of the form
\begin{equation}
	\mc{M}(q(t))\ddot{q}(t) + \mc{C}(q(t), \dot{q}(t))\dot{q}(t) = \tau_2(t),
\end{equation}
where $q^\top = \begin{bmatrix} q_1 & q_2 & q_4 \end{bmatrix}$ are the angular positions, $\tau_2$ is the input torque, $\mc{M}$ and $\mc{C}$ are the inertia and Coriolis matrices.

In the used configuration of the CMG, the goal is to control the position of gimbal $A$ by actuating gimbal $C$ through input torque $\tau_2$. 
The driving factor in this setting is the velocity of the disk $D$, which directly relates to the amount of gyroscopic torque that can be exerted on gimbal $A$. In \cite{abbas2014lpv}, it is shown that local linear approximations describe the nonlinear dynamics accurately. The aggregated collection of these local approximations is described by the following representation
\begin{subequations}
	\begin{align}
		\dot{x}(t) &= A(\dot{q}_1(t))x(t) + Bu(t), \\
		y(t) &= q_4(t),
	\end{align}
\end{subequations}
where $x^\top = \begin{bmatrix} q_4 & \dot{q}_2 & \dot{q}_4 \end{bmatrix}$ is the state, $u = \tau_2$ the input and $y = q_4$ the output. The $A$ matrix depends on the velocity of the disk, which can range anywhere in $\dot{q}_1 \in [30, 50]$ rad/s.

The local description of the behavior is in line with the availability of measurement data and the considered control synthesis techniques in the sequel. Furthermore, the dependence of the system on the disk velocity makes the LPV framework a suitable choice for modeling and control.

\subsection{LPV systems}
Consider a single-input single-output, continuous-time (CT) LPV system. The LPV state-space representation
\begin{align}
	\label{eqn:LPVss}
	G_p:
	\begin{cases}
		\dot{x}(t) &= A(p(t))x(t) + B(p(t))u(t),\\
		y(t) &= C(p(t))x(t) + D(p(t))u(t),
	\end{cases}
\end{align}
is adopted to represent the system, see also \cite{toth2010modeling}. Here, $x : \mbb{R} \rightarrow \mbb{X} \subseteq \mbb{R}^{n_\mrm{x}}$ denotes the state variable, $u : \mbb{R} \rightarrow \mbb{U} \subseteq \mbb{R}$ is the input signal, $y: \mbb{R} \rightarrow \mbb{Y} \subseteq \mbb{R}$ is the output signal and $p : \mbb{R} \rightarrow \mbb{P} \subseteq \mbb{R}^{n_\mrm{p}}$ the scheduling variable.

With a slight abuse of notation introduce
\begin{equation}
\label{eqn:LPVssShort}
	G_p = 
	\begin{pmatrix}
		\begin{array}{c|c}
			A(p) & B(p) \\ \hline
			C(p) & D(p)
		\end{array}
	\end{pmatrix}
\end{equation}
representing the LPV system with state-space form \eqref{eqn:LPVss}. If $D^{-1}(p)$ is well-defined for all $p \in \mbb{P}$, then the LPV system $G_p$ has an inverse operator
\begin{equation}
	G_p^{-1} \!=\! 
	\begin{pmatrix}
		\begin{array}{c|c}
			A(p) + B(p)D^{-1}(p)C(p) & B(p)D^{-1}(p) \\ \hline
			D^{-1}(p)C(p) & D^{-1}(p)
		\end{array}
	\end{pmatrix}, 
\end{equation}
such that $G_pG_p^{-1} = G_p^{-1}G_p = 1$ for all $p \in \mbb{P}$.

If the scheduling signal $p(t) \equiv \mrm{p}$ is constant, the scheduling-dependent matrices in \eqref{eqn:LPVssShort} become time-invariant, i.e.,
\begin{equation}
\label{eqn:LPVssFrozen}
	G_\mrm{p} = 
	\begin{pmatrix}
		\begin{array}{c|c}
			A(\mrm{p}) & B(\mrm{p}) \\ \hline
			C(\mrm{p}) & D(\mrm{p})
		\end{array}
	\end{pmatrix}
\end{equation}
represents an LTI system for constant scheduling. For a given $\mrm{p} \in \mbb{P}$, \eqref{eqn:LPVssFrozen} describes the local behavior of \eqref{eqn:LPVss}. Hence, \eqref{eqn:LPVssFrozen} is referred to as the frozen behavior of \eqref{eqn:LPVss}. Taking the Laplace transform of \eqref{eqn:LPVssFrozen} with zero initial conditions results in
\begin{equation}
\label{eqn:LPVIOFrozen}
	\hat{y}(s) = \left(C(\mrm{p})(sI - A(\mrm{p}))^{-1}B(\mrm{p}) + D(\mrm{p})\right)\hat{u}(s),
\end{equation}
where $G_\mrm{p}(s) = C(\mrm{p})(sI - A(\mrm{p}))^{-1}B(\mrm{p}) + D(\mrm{p})$ and $s$ is the Laplace variable. The frozen behavior \eqref{eqn:LPVssFrozen} also has a corresponding Fourier transform
\begin{equation}
\label{eqn:fFRF}
	Y(i\omega) = G_\mrm{p}(i\omega)U(i\omega),
\end{equation}
where $i$ is the complex unit, $\omega \in \mbb{R}$ is the frequency and $G_\mrm{p}(i\omega)$ represents the frozen Frequency Response Function (fFRF) of \eqref{eqn:LPVss} for every constant $p(t) \equiv \mrm{p} \in \mbb{P}$ \cite{schoukens2019frequency}.

\begin{figure}[t]
	\centering
	\includegraphics[scale=1.00]{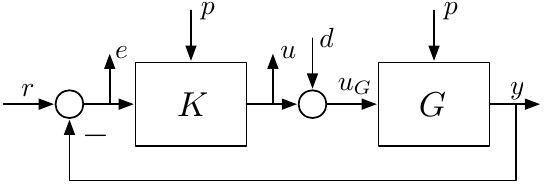}
	\caption{Typical 1 DOF feedback interconnection, including 4-block shaping problems, depending on the scheduling signal.}
	\label{fig:internal_stability}
	\vspace{-3ex}
\end{figure}

\subsection{Problem statement}
\label{subsection:problem_statement}
The problem addressed in this paper is to design an LPV controller directly from fFRF measurement data obtained from the considered CMG. We denote the data $\mc{D}_{N, \mrm{p}_\tau} = \{ G_\mrm{p}(i\omega_k), \mrm{p}_\tau \}_{k=1}^{N}$, obtained at the set of operating points $\mc{P} = \{ \mrm{p}_\tau \}_{\tau=1}^{N_\mrm{loc}} \subset \mbb{P}$. We assume the frequencies are sufficiently dense such that it suffices to check a finite number of discrete points to draw conclusions on the underlying continuous curve. Consider the feedback interconnection in Figure \ref{fig:internal_stability}. The objective is to design a controller $K_p$ such that the following requirements are satisfied.
\begin{enumerate}[{R1)},left=\parindent]
	\item {\label{requirement1}} The closed-loop system in Figure \ref{fig:internal_stability} is internally stable in the local sense for all $p(t) \equiv \mrm{p} \in \mc{P}$.
	\item {\label{requirement2}} The performance channels of the closed-loop system are bounded in the local $\mc{H}_\infty$-norm sense for all $\mrm{p} \in \mc{P}$.
\end{enumerate}

In the next section, a rational controller parameterization is introduced that allows for a specific formulation of internal stability. This forms the basis to develop analysis conditions for internal stability and $\mc{H}_\infty$-performance. The theory is first formulated for $\mrm{p} \in \mbb{P}$ for the sake of generality. This also ensures \ref{requirement1} and \ref{requirement2} for $\mrm{p} \in \mc{P}$. 

\section{Stability and performance analysis}
\label{section:analysis_conditions}
In this section, we develop local LPV stability and performance analysis conditions. These results form the basis for a data-driven synthesis procedure. First, a continuous frequency spectrum $\Omega = \{\mbb{R} \cup \{ \infty \} \}$ is considered, which will be restricted later to a finite frequency grid $\Omega_N = \{ \omega_k \}_{k=1}^{N}$ corresponding to $\mc{D}_{N,\mrm{p}_\tau}$. 

\subsection{Stability}
The selection of input-output (IO) pairs in Figure \ref{fig:internal_stability} corresponds to the problem of internal stability \cite[Chapter 3]{doyle1992feedback}. For a fixed $\mrm{p} \in \mbb{P}$, we define the IO map $T(G_\mrm{p}, K_\mrm{p}) : (r, -d) \mapsto (e, u)$ in Figure \ref{fig:internal_stability} by 
\begin{equation}
\label{eqn:four-block}
	T(G_\mrm{p}, K_\mrm{p}) =
	\begin{bmatrix}
		S_\mrm{p} & S_\mrm{p}G_\mrm{p} \\
		K_\mrm{p}S_\mrm{p} & T_\mrm{p}		
	\end{bmatrix}, 	
\end{equation}
with $S_\mrm{p} = (1+G_\mrm{p}K_\mrm{p})^{-1}$ and $T_\mrm{p} = 1-S_\mrm{p}$. If $G_\mrm{p}, K_\mrm{p} \in \RHinf$, then $T(G_\mrm{p}, K_\mrm{p})$ is internally stable if all elements in the IO map $T(G_\mrm{p}, K_\mrm{p})$, defined by \eqref{eqn:four-block}, are stable. This is implied by $S_\mrm{p} \in \RHinf$ \cite[Chapter 3]{doyle1992feedback}. If $T(G_\mrm{p}, K_\mrm{p}) \in \RHinf$ holds for all $\mrm{p} \in \mbb{P}$, then the closed-loop LPV system is called locally internally stable. 
Internal stability is imporant to prevent hidden pole-zero cancellations. To assess internal stability for unstable $G_\mrm{p}$ or $K_\mrm{p}$, introduce the factorization
\begin{equation}
\label{eqn:Gparameterization}
	G_\mrm{p} = N_{G_\mrm{p}}D_{G_\mrm{p}}^{-1}, \quad \{ N_{G_\mrm{p}}, D_{G_\mrm{p}}\} \in \RHinf.
\end{equation}
The two transfer functions $\{ N_{G_\mrm{p}}, D_{G_\mrm{p}}\}$ are a coprime factorization over $\RHinf$ if there exist two other transfer functions $\{ X_\mrm{p}, Y_\mrm{p}\} \in \RHinf$ such that they satisfy the B\'{e}zout identity
\begin{equation}
\label{eqn:bezout}
	N_{G_\mrm{p}} X_\mrm{p} + D_{G_\mrm{p}}Y_\mrm{p} = 1.
\end{equation}
Consequently, $\{X_\mrm{p}Q, Y_\mrm{p}Q\}$ are coprime iff $Q, Q^{-1} \in \RHinf$.
Correspondingly, $K_\mrm{p}$ admits the coprime factorization
\begin{equation}
\label{eqn:Kparameterization}
	K_\mrm{p} = N_{K_\mrm{p}}D_{K_\mrm{p}}^{-1}, \quad \{ N_{K_\mrm{p}}, D_{K_\mrm{p}}\} \in \RHinf.
\end{equation}
Using these representations, \eqref{eqn:four-block} can be written as
\begin{equation}
\label{eqn:four-block-coprime}
	T(G_\mrm{p}, K_\mrm{p}) = 
	D_{\mrm{p}}^{-1}
	\begin{bmatrix}
		D_{G_\mrm{p}}D_{K_\mrm{p}} & N_{G_\mrm{p}}D_{K_\mrm{p}} \\
		D_{G_\mrm{p}}N_{K_\mrm{p}} & N_{G_\mrm{p}}N_{K_\mrm{p}}
	\end{bmatrix},
\end{equation}
with characteristic equation 
\begin{equation}
\label{eqn:characteristicEquation}
	D_{\mrm{p}} = D_{G_\mrm{p}}D_{K_\mrm{p}} + N_{G_\mrm{p}}N_{K_\mrm{p}}.
\end{equation}
The feedback system is internally stable if and only if $D_{\mrm{p}}^{-1} \in \RHinf$. If we set $N_{K_\mrm{p}} = X_\mrm{p}$ and $D_{K_\mrm{p}} = Y_\mrm{p}$, then the characteristic equation \eqref{eqn:characteristicEquation} equals the B\'{e}zout identity \eqref{eqn:bezout}, thus the feedback system is internally stable as $D_\mrm{p}^{-1} = 1$ and the rest of the terms are stable by design in \eqref{eqn:four-block-coprime}.
 Similarly, the closed-loop LPV system is called locally internally stable if these conditions hold for all $\mrm{p} \in \mbb{P}$.

For the transfer $w \mapsto z$, with $w \in \{r, d\}$ and $z \in \{e,u\}$, let
\begin{equation}
\label{eqn:TSISO_definition}
	T_{z,w}(G_\mrm{p}, K_\mrm{p}) = N_{\mrm{p}}D_{\mrm{p}}^{-1},
\end{equation}
with $\{ N_{\mrm{p}}, D_{\mrm{p}} \} \in \RHinf$ and $T_{z,w}(G_\mrm{p}, K_\mrm{p}) \in \RHinf$, defines the corresponding SISO element of \eqref{eqn:four-block-coprime}. For example, $T_{r,e}(G_\mrm{p}, K_\mrm{p}) = N_{\mrm{p}}D_{\mrm{p}}^{-1}$ with $N_{\mrm{p}} = D_{G_\mrm{p}}D_{K_\mrm{p}}$ defines the sensitivity $S_\mrm{p}$ in \eqref{eqn:four-block} and \eqref{eqn:four-block-coprime}.

The following theorem presents analysis conditions to verify internal stability of a closed-loop system locally, given the plant and controller only. As a special case, \cite[Theorem 1]{rantzer1994} is recovered. Here, the idea of coprime factorizations over $\RHinf$ is used to allow for unstable plants or controllers, while also extending the result to the class of LPV systems.
\begin{theorem}
\label{thm:rantzer1994_coprime}
Let $G_\mrm{p}$ and $K_\mrm{p}$ be as defined in \eqref{eqn:Gparameterization} and \eqref{eqn:Kparameterization}, respectively, and let $D_{\mrm{p}} \in \RHinf$ be as defined in \eqref{eqn:characteristicEquation}. Then the following conditions are equivalent. For all $\mrm{p} \in \mbb{P}$
\begin{enumerate}[{\ref{thm:rantzer1994_coprime}\alph*)},left=\parindent]
	\item {\label{rantzer1994_coprime:stability_coprime}} $D_{\mrm{p}}^{-1} \in \RHinf$.
	\item {\label{rantzer1994_coprime:nonzero}} $D_{\mrm{p}}(s) \neq 0, \, \forall s \in \mbb{C}_+ \cup \mbb{C}_{0} \cup \{ \infty \}$.
	\item {\label{rantzer1994_coprime:posreal_rhinf}} There exists a multiplier $\alpha_\mrm{p} \in \RHinf$ such that 
		\vspace{-1ex}
		\begin{equation*}
		\vspace{-1ex}
			\Re \{ D_{\mrm{p}}(i\omega)\alpha_\mrm{p}(i\omega) \} > 0, \, \forall \omega \in \Omega. 
		\end{equation*} 		
\end{enumerate}
\end{theorem}
\begin{proof}
	For a proof of equivalence between \ref{rantzer1994_coprime:stability_coprime} and \ref{rantzer1994_coprime:nonzero}, see \cite[Chapter 3]{doyle1992feedback}. Regarding the equivalence between \ref{rantzer1994_coprime:stability_coprime} and \ref{rantzer1994_coprime:posreal_rhinf} for all $\mrm{p} \in \mbb{P}$, note the following reasoning:
	
	$(\Rightarrow)$ Assume \ref{rantzer1994_coprime:stability_coprime} and let $Q = D_{\mrm{p}}^{-1}$. This implies that the B\'ezout identity \eqref{eqn:bezout} is satisfied for $X_\mrm{p} = N_{K_\mrm{p}}Q$ and $Y_\mrm{p} = D_{K_\mrm{p}}Q$. Hence, \ref{rantzer1994_coprime:posreal_rhinf} is satisfied by setting $\alpha_\mrm{p} = Q$ because $\Re \{ N_{G_\mrm{p}} X_\mrm{p} + D_{G_\mrm{p}}Y_\mrm{p} \} = 1$ for all  $\omega \in \Omega$.
	
	$(\Leftarrow)$ Assume \ref{rantzer1994_coprime:posreal_rhinf} and let $V = D_{\mrm{p}}\alpha_\mrm{p}$. Note that $V, V^{-1} \in \RHinf$ because \ref{rantzer1994_coprime:posreal_rhinf} implies that $D_{\mrm{p}}\alpha_\mrm{p}$ is bi-proper and has no right half-plane (RHP) zeros. Then $D_{\mrm{p}} = V\alpha_\mrm{p}^{-1}$ satisfies the B\'{e}zout identity \eqref{eqn:bezout}, therefore $D_{\mrm{p}}^{-1} \in \RHinf$. Thus \ref{rantzer1994_coprime:posreal_rhinf} implies \ref{rantzer1994_coprime:stability_coprime} and consequently \ref{rantzer1994_coprime:nonzero}. This completes the proof.	
\end{proof}
\begin{remark}
\vspace{-1ex}
	A direct result of Theorem \ref{thm:rantzer1994_coprime} is that $\alpha_\mrm{p}^{-1} \in \RHinf$. This is easy to prove because
	\begin{enumerate}[i)]
		\item There does not exist a strictly proper $\alpha_\mrm{p} \in \RHinf$ such that \ref{rantzer1994_coprime:posreal_rhinf} holds. Indeed \ref{rantzer1994_coprime:posreal_rhinf} is violated at $\omega = \infty$.
		\item There does not exist an $\alpha_\mrm{p} \in \RHinf$ with $\alpha_\mrm{p}^{-1} \notin \RHinf$ such that \ref{rantzer1994_coprime:posreal_rhinf} holds. This can be seen as $\alpha_\mrm{p}^{-1} \notin \RHinf$ implies that there exists some RHP zero $s_0$ such that $\alpha_\mrm{p}(s_0) = 0$. Consequently, there exists some frequency $\omega_0$ such that $\Re \{ D_{\mrm{p}}(i\omega_0)\alpha_\mrm{p}(i\omega_0) \} < 0$ and \ref{rantzer1994_coprime:posreal_rhinf} is violated. 
	\end{enumerate} 
\end{remark}

Theorem \ref{thm:rantzer1994_coprime} gives an analysis condition that provides a local stability result for the closed-loop system if instead of a parametric model, $N_{G_\mrm{p}}$ and $D_{G_\mrm{p}}$ are only given in terms of local frequency-domain data. 
The next subsection presents the extension towards a performance analysis condition.

\subsection{Performance}
In this subsection, analysis conditions to assess locally the $\mc{H}_\infty$-performance of an LPV system, given the plant and controller only, are presented. This constitutes contribution \ref{Contribution:4}. To derive performance analysis conditions, the main loop theorem is of importance and is presented first.

Consider the transfer function $T_{z,w}(G_\mrm{p}, K_\mrm{p}) \in \RHinf$ of interest in Figure \ref{fig:PKForm}, such that $w \mapsto z : T_{z,w}(G_\mrm{p}, K_\mrm{p})$, and let ${ \hat{\Delta} \in \mathbf{B\hat{\Delta}} }$, with
\begin{equation}
\mathbf{B\hat{\Delta}} := \Set{\hat{\Delta} \in \RHinf}{\lvert \hat{\Delta}(i\omega) \rvert < 1, \, \forall \omega \in \Omega}
\end{equation}
a fictitious uncertainty, represent the $\mc{H}_\infty$-performance criterion. Then, the $\mc{H}_\infty$-performance of the system in Figure \ref{fig:PKForm} is equivalent to Figure \ref{fig:MLBD} \cite[Theorem 8.7]{skogestad2001MFC}. This is stated in terms of the following theorem, where the weighting filter $W_T$ is introduced to specify the frequency-dependent design requirements on the map $w \mapsto z$.
\begin{theorem}[Main loop theorem]
\label{thm:main_loop}
	Let $W_T \in \RHinf$ and $T_{z,w}(G_\mrm{p}, K_\mrm{p})$ be defined as in \eqref{eqn:TSISO_definition}. The following statements are equivalent. For all $\mrm{p} \in \mbb{P}$
	\begin{enumerate}[{\ref{thm:main_loop}\alph*)},left=\parindent]
		\item {\label{thm:main_loop_a}} $\underset{\omega \in \Omega}{\sup} \, \lvert  W_T(i\omega) T_{z,w}(G_\mrm{p}, K_\mrm{p})(i\omega) \rvert \leq \gamma$.
		\item {\label{thm:main_loop_b}} $ \begin{aligned}[t] 1 - \gamma^{-1}W_T(i\omega) T_{z,w}(G_\mrm{p}, K_\mrm{p})(i\omega)\hat{\Delta}(i\omega) \neq 0, \\ \forall \omega \in \Omega, \, \forall \hat{\Delta} \in \mathbf{B\hat{\Delta}}. \end{aligned}$
	\end{enumerate}
\end{theorem}
Theorem \ref{thm:main_loop} is a special case of \cite[Theorem 11.7]{zhou1996robust}.
\begin{remark}
	By Theorem \ref{thm:main_loop}, nominal performance can be seen as a special case of robust stability, where a fictitious uncertainty is connected to the performance channel, see Figure \ref{fig:MLBD}.
\end{remark}
\begin{figure}[]
	\centering
	\begin{subfigure}[b]{0.45\columnwidth}
		\centering
		\includegraphics[scale=1.00]{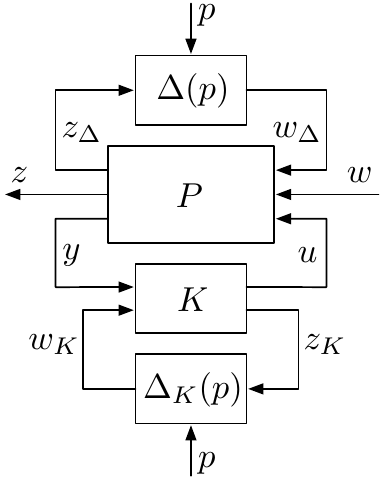}
		\caption{}
		\label{fig:PKForm}
	\end{subfigure}
	\hfill
	\begin{subfigure}[b]{0.45\columnwidth}
		\centering
		\includegraphics[scale=1.00]{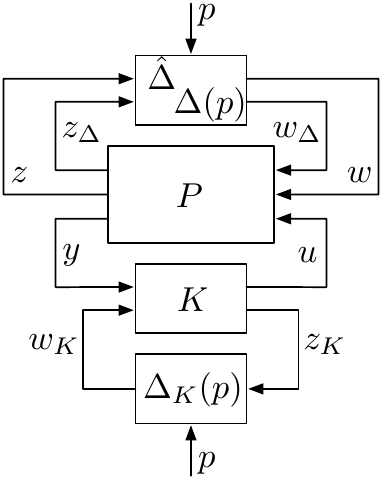}
		\caption{}
		\label{fig:MLBD}
	\end{subfigure}
	\caption{Generalized LPV plant (a); and performance of the SISO closed-loop map $w \mapsto z$ (b).}
	\vspace{-3ex}
\end{figure}

In the data-driven setting, the absence of a parametric model of $T_{z,w}(G_\mrm{p}, K_\mrm{p})$ makes it difficult to turn \ref{thm:main_loop_b} into a convex constraint as it is generally done in LPV synthesis approaches for gain-scheduling \cite{hoffmann2015survey}. Hence, in that case \ref{thm:main_loop_b} is needed to be evaluated for an infinite set of realizations of the fictitious uncertainty $\hat{\Delta}$, for example, as in \cite{van2018frequency}. The contribution in this paper is to utilize Theorem \ref{thm:rantzer1994_coprime} together with Theorem \ref{thm:main_loop} to derive a single condition to analyze both stability and performance without the need to sample $\hat{\Delta}$.
\begin{theorem}
\label{thm:performance_analysis}
	Let $W_T \in \RHinf$ and $T_{z,w}(G_\mrm{p}, K_\mrm{p})$ be defined as in \eqref{eqn:TSISO_definition}. Requirements \ref{requirement1} and \ref{requirement2} are satisfied if and only if there exists a multiplier $\alpha_\mrm{p} \in \RHinf$ such that
	\begin{equation}
	\begin{split}
	\label{eqn:performance_analysis}
		\Re \{ (D_{\mrm{p}}(i\omega) - \gamma^{-1}\lvert W_T(i\omega)N_{\mrm{p}}(i\omega) \rvert)\alpha_\mrm{p}(i\omega) \} > 0, \\
		\forall \omega \in \Omega, \, \forall \mrm{p} \in \mbb{P}.
	\end{split}
	\end{equation} 
\end{theorem}
\begin{proof}
	Requirement \ref{requirement2} can be equivalently stated using Theorem \ref{thm:main_loop}, Condition \ref{thm:main_loop_b}, i.e.,
	\begin{equation}
	\begin{split}
	\label{eqn:proof_performance_analysis_1}
			1 - \gamma^{-1}W_T(i\omega)T_{z,w}(G_\mrm{p}, K_\mrm{p})(i\omega)\hat{\Delta}(i\omega) \neq 0, \\ 
			\forall \omega \in \Omega, \, \forall \mrm{p} \in \mbb{P}, \, \forall \hat{\Delta} \in \mathbf{B\hat{\Delta}}.
	\end{split}
	\end{equation}
	As $D_\mrm{p} \in \RHinf$, $D_\mrm{p}(i\omega)\neq 0,$ $\forall \omega \in \Omega$ and by multiplying \eqref{eqn:proof_performance_analysis_1} with it, the resulting non-singularity condition is:
	\begin{equation}
	\begin{split}
	\label{eqn:proof_performance_analysis_1b}
			D_{\mrm{p}}(i\omega) - \gamma^{-1} W_T(i\omega)N_{\mrm{p}}(i\omega) \hat{\Delta}(i\omega)\neq 0, \\ 
			\forall \omega \in \Omega, \, \forall \mrm{p} \in \mbb{P}, \, \forall \hat{\Delta} \in \mathbf{B\hat{\Delta}}.
	\end{split}
	\end{equation}
	Based on a homotopy argument, \eqref{eqn:proof_performance_analysis_1b} corresponds to Condition 1b) in Theorem \ref{thm:rantzer1994_coprime}, which through 1c) is equivalent with
	\begin{equation}
	\label{eqn:proof_performance_analysis_3}
	\begin{split}
		\Re \{ (D_{\mrm{p}}(i\omega) - \gamma^{-1} W_T(i\omega)N_{\mrm{p}}(i\omega) \hat{\Delta}(i\omega))\alpha_\mrm{p}(i\omega) \} \! > \! 0, \\ 
		\forall \omega \in \Omega, \, \forall \mrm{p} \in \mbb{P}, \, \hat{\Delta} \in \mathbf{B\hat{\Delta}}.
	\end{split}
	\end{equation}
	When $\hat{\Delta} = 0 \in \mathbf{B\hat{\Delta}}$, \eqref{eqn:proof_performance_analysis_3} reduces to $\Re \{ D_{\mrm{p}}(i\omega)\alpha_\mrm{p}(i\omega) \} > 0$, which is the same as Condition \ref{rantzer1994_coprime:posreal_rhinf} in Theorem \ref{thm:rantzer1994_coprime}, hence \eqref{eqn:proof_performance_analysis_3} implies requirement \ref{requirement1}.
	
	Let $1\geq \epsilon>0$ and consider  \eqref{eqn:proof_performance_analysis_3} on
	\begin{equation}
\mathbf{B}_\epsilon \mathbf{\hat{\Delta}} := \Set{\hat{\Delta} \in \RHinf}{\lvert \hat{\Delta}(i\omega) \rvert \leq 1-\epsilon, \, \forall \omega \in \Omega},
\end{equation}
	which is the scaled closed uncertainty ball contained in $\mathbf{B\hat{\Delta}}$. Since any $\hat{\Delta} \in \mathbf{B}_\epsilon \mathbf{\hat{\Delta}}$ represents a rotation and contraction in the complex plane, it is necessary and sufficient to check \eqref{eqn:proof_performance_analysis_3} on the boundary only, i.e., for $\hat{\Delta} \in \partial  \mathbf{B}_\epsilon \mathbf{\hat{\Delta}}$, with $\lvert \hat \Delta(i\omega) \rvert = 1-\epsilon$, $\forall \omega \in \Omega$. 
Note that, in \eqref{eqn:proof_performance_analysis_3}, $W_T(i\omega)N_{\mrm{p}}(i\omega)$ only represents complex scaling of this ball which is centered at $D_{\mrm{p}}(i\omega)$. Hence, \eqref{eqn:proof_performance_analysis_3} restricted on $\mathbf{B}_\epsilon \mathbf{\hat{\Delta}}$ is equivalent with 
	\begin{multline}
		\Re \{ (D_{\mrm{p}}(i\omega) - \gamma^{-1}(1-\epsilon)\lvert W_T(i\omega)N_{\mrm{p}}(i\omega) \rvert)\alpha_\mrm{p}(i\omega) \} > 0, \\
		\forall \omega \in \Omega, \, \forall \mrm{p} \in \mbb{P}.
	\label{scaled_eq}
	\end{multline} 
	This means that if \eqref{scaled_eq} holds, then violation of \eqref{eqn:proof_performance_analysis_3} can only happen in $\mathbf{B\hat{\Delta}} \setminus \mathbf{B}_\epsilon \mathbf{\hat{\Delta}}$. As \eqref{scaled_eq} is continuous in $\epsilon$, by taking the limit $\epsilon \rightarrow 0$, 
	$\mathbf{B\hat{\Delta}} \setminus \mathbf{B}_\epsilon \mathbf{\hat{\Delta}}\rightarrow \emptyset$ and we obtain that \eqref{eqn:performance_analysis} is equivalent with \eqref{eqn:proof_performance_analysis_3}.
\end{proof}

Theorem \ref{thm:performance_analysis} states that the performance condition \ref{thm:main_loop_a} is satisfied if and only if for each frequency $\omega \in \Omega$ and scheduling value $\mrm{p} \in \mbb{P}$ the disks with radius $\gamma^{-1}\lvert W_T N_{\mrm{p}} \rvert$, centered at $D_{\mrm{p}}$, do not include the origin. This holds if there exists $\alpha_\mrm{p} \in \RHinf$, representing for each frequency a line passing through the origin, that does not intersect with the disks, see Figure \ref{fig:performance_illustration}. The analysis condition is especially useful as it provides a local stability and performance result given a controller and the data $\mc{D}_{N,\mrm{p}_\tau}$. 

If the fFRFs are subject to model uncertainty, robust stability and performance have to be taken into account \cite{karimi2018robust}.
\begin{figure}[]
	\centering
	\includegraphics[scale=1.00]{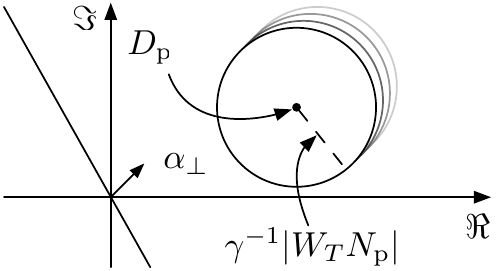}
	\caption{Illustration of stability and $\mc{H}_\infty$-performance. The transfer function $\alpha_\mrm{p}$ represents, for each frequency, a line passing through the origin. If this line does not intersect with the disks $D_{\mrm{p}} - \gamma^{-1}\lvert W_TN_{\mrm{p}} \lvert$, then the disks exclude the origin and \eqref{eqn:proof_performance_analysis_1} must hold.}
	\label{fig:performance_illustration}
	\vspace{-3ex}
\end{figure}
\subsection{Synthesis}
It turns out that it is possible to give an equivalent formulation of Theorem \ref{thm:performance_analysis} which enables controller synthesis.
\begin{theorem}
\label{thm:synthesis}
	Given $G_\mrm{p} = N_{G_\mrm{p}}D_{G_\mrm{p}}^{-1}$, with $\{ N_{G_\mrm{p}}, D_{G_\mrm{p}} \} \in \RHinf$ coprime, as defined in \eqref{eqn:Gparameterization}, and a weighting filter $W_T \in \RHinf$, the following statements are equivalent.
	\begin{enumerate}[{\ref{thm:synthesis}\alph*)},left=\parindent]
		\item {\label{thm:karimi_a}} There exists a proper rational controller $K_\mrm{p}$ that achieves internal stability and performance as defined in requirements \ref{requirement1} and \ref{requirement2}, respectively.
		\item {\label{thm:karimi_b}} There exists a controller $K_\mrm{p} = N_{K_\mrm{p}}D_{K_\mrm{p}}^{-1}$, with $\{ N_{K_\mrm{p}}, D_{K_\mrm{p}} \} \in \RHinf$, as defined in \eqref{eqn:Kparameterization}, such that
			\begin{equation}
			\begin{split}
			\label{eqn:synthesis_constraint}
				\Re \{ D_{\mrm{p}}(i\omega) \} > \gamma^{-1}\lvert W_T(i\omega)N_{\mrm{p}}(i\omega) \rvert, \\
				\forall \omega \in \Omega,  \, \forall \mrm{p} \in \mbb{P}.
			\end{split}
			\end{equation}
	\end{enumerate}
\end{theorem}
\begin{proof}
	$(\Rightarrow)$ Assume $K_\mrm{p} = \tilde{N}_{K_\mrm{p}} \tilde{D}_{K_\mrm{p}}^{-1}$ satisfies \ref{thm:karimi_a}. Then, by Theorem \ref{thm:performance_analysis}, there exists an $\alpha_\mrm{p} \in \RHinf$ such that \eqref{eqn:performance_analysis} holds. Choosing $N_{K_\mrm{p}} = \tilde{N}_{K_\mrm{p}}\alpha_\mrm{p}$, $D_{K_\mrm{p}} = \tilde{D}_{K_\mrm{p}}\alpha_\mrm{p}$ results in $K_\mrm{p} = N_{K_\mrm{p}}D_{K_\mrm{p}}^{-1} = \tilde{N}_{K_\mrm{p}}\tilde{D}_{K_\mrm{p}}^{-1}$ and consequently \ref{thm:karimi_b} holds.
	
	$(\Leftarrow)$ Assume \ref{thm:karimi_b} holds. Because $D_\mrm{p} \in \RHinf$ and $D_{\mrm{p}}(i\omega)$ is positive for all $\omega \in \Omega$, $\{N_{K_\mrm{p}}, D_{K_\mrm{p}}\}$ form B\'ezout factors for $\{ N_{G_\mrm{p}}, D_{G_\mrm{p}}\}$. Thus by Theorem \ref{thm:rantzer1994_coprime}, $D_{\mrm{p}}^{-1} \in \RHinf$ and $K_\mrm{p}$ internally stabilizes $G_\mrm{p}$ and \ref{requirement1} holds. By Theorem \ref{thm:performance_analysis}, requirement \ref{requirement2} holds. This completes the proof.
\end{proof}

Theorem \ref{thm:synthesis} presents a local $\mc{H}_\infty$-optimal controller synthesis condition given only data $\mc{D}_{N,\mrm{p}_\tau}$. This is further developed in Section \ref{section:synthesis}, where an optimization problem is formulated and the controller parameterization is discussed.
\begin{remark}
		Theorem \ref{thm:synthesis} shows that the multiplier $\alpha_\mrm{p}$ can be absorbed into the controller as $\gamma^{-1} \lvert W_T(i\omega)N_{\mrm{p}}(i\omega) \alpha_\mrm{p}(i\omega) \rvert \Rightarrow \Re \{ \gamma^{-1} \lvert W_T(i\omega)N_{\mrm{p}}(i\omega) \rvert \alpha_\mrm{p}(i\omega) \}$. Note that the absorbed multiplier changes the considered $N_{\mrm{p}}$ and $D_{\mrm{p}}$, but $\alpha_\mrm{p}$ cancels out when $K_\mrm{p} = N_{K_\mrm{p}}D_{K_\mrm{p}}^{-1}$ is computed. The price to be paid for this absorption is the increased order of $N_{K_\mrm{p}}$ and $D_{K_\mrm{p}}$.
\end{remark}
\begin{remark}
	\cite[Theorem 1]{karimi2018robust} is recovered in the special case when the plant and controller are LTI . 
\end{remark}
\section{Controller synthesis}
\label{section:synthesis}
In this section, we build upon the stability and performance analysis and synthesis conditions derived in Section \ref{section:analysis_conditions} by developing a procedure to synthesize LPV controllers. This forms Contribution C1). First, an optimization problem is set up in Section \ref{subsection:synthesisOptimization} that characterizes the synthesis problem based on Theorem \ref{thm:synthesis}. This is followed by a discussion on the controller parameterization in Section \ref{subsection:Kparameterization} and implementation aspects in Section \ref{subsection:Kimplementation}.

\subsection{Controller synthesis}
\label{subsection:synthesisOptimization}
Given the data $\mc{D}_{N, \mrm{p}_\tau}$ and a controller parameterization $K_\mrm{p} = N_{K_\mrm{p}}D_{K_\mrm{p}}^{-1}$, given in the Section \ref{subsection:Kparameterization}, an optimization problem is formulated satisfying requirements \ref{requirement1} and \ref{requirement2}.
\begin{equation}
\label{eqn:synthesis}
\begin{aligned}
	& \underset{\theta, \gamma}{\text{min}}
	& & \gamma \\
	& \text{s.t.}
	& & \gamma \Re \{ D_{\mrm{p}}(i\omega, \theta) \} > \lvert W_T(i\omega)N_{\mrm{p}}(i\omega, \theta) \rvert \\
	& & & \forall \omega \in \Omega, \, \mrm{p} \in \mc{P}
\end{aligned}
\end{equation}
where $\theta$ are the controller parameters. 

The optimization problem \eqref{eqn:synthesis} is in general non-convex. However, through a linear parameterization of the controller, \eqref{eqn:synthesis} becomes a quasi-convex optimization problem in the controller parameters $\theta$ and the performance indicator $\gamma$. To solve the quasi-convex program, a bisection algorithm over $\gamma$ is utilized. This results in an iterative approach, where for every fixed value of $\gamma$, a second-order cone program is solved. 

To provide stability and performance guarantees, the constraints in \eqref{eqn:synthesis} need to be satisfied for all $\omega \in \Omega$, which is an infinite set, leading to a semi-infinite program. One solution is to solve \eqref{eqn:synthesis} for a finite set of frequencies $\Omega_N = \{ \omega_k \}_{k=1}^{N} \subset \Omega$. The frequency set can be chosen randomly, according to the scenario approach \cite{calafiore2006scenario}. This allows for the computation of confidence bounds on the constraints. In the data-driven setting this choice is spared from the user as the data is only available at a pre-specified set of frequency points. Either of these methods result in a quasi-convex second-order cone program and can be solved as described above.

\subsection{Controller parameterization}
\label{subsection:Kparameterization}

An orthonormal basis function (OBF)-based representation \cite{toth2010modeling} is a natural choice to parameterize the controller factors
\begin{subequations}
\begin{align}
\label{eqn:Nparameterization}
	N_{K_\mrm{p}}(s) &= {\textstyle \sum_{i=0}^{n_N}} \, w_i(\mrm{p})\phi_i(s), \\
	\label{eqn:Dparameterization}
	D_{K_\mrm{p}}(s) &= {\textstyle \sum_{i=0}^{n_D}} \, v_i(\mrm{p})\varphi_i(s),
\end{align}
\end{subequations}
Here, $\{ \phi_i \}_{i=0}^{n_N}$ and $\{ \varphi_i \}_{i=0}^{n_D}$ with $\phi_0 = \varphi_0 = 1$ and $n_D \geq n_N$ are the sequence of basis functions, with coefficient functions 
\begin{align}
	w_i(\mrm{p}) = {\textstyle \sum_{\ell=1}^m} \, \breve{w}_{i}^{\ell} \psi_{\ell}(\mrm{p}), && v_i(\mrm{p}) = {\textstyle \sum_{\ell=1}^m} \, \breve{v}_{i}^{\ell} \psi_{\ell}(\mrm{p}).
\end{align}
Here, the coefficient functions are formed through a chosen functional dependence, e.g., affine, polynomial or rational, characterized by the basis functions $\{ \psi_\ell \}_{\ell=1}^{m}$. See \cite[Chapter 9.2]{toth2010modeling} for an overview of OBF based LPV model structures. The OBF controller parameterization enables tuning of both the poles and zeros of the controller, in contrast to previous data-driven frequency-domain LPV tuning methods \cite{kunze2007gain,karimi2013hinf,bloemers2019towards_lpv_synthesis}. Additional controller requirements are discussed in \cite{bloemers2021FrfLpvSyn}.
\begin{figure}[]
	\centering
	\includegraphics[scale=1.00]{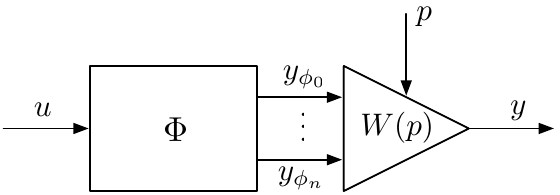}
	\caption{Input-output graph of the Wiener LPV OBF structure.}
	\label{fig:WienerOBF}
	\vspace{-3ex}
\end{figure}%

Local aspects of \eqref{eqn:Nparameterization}-\eqref{eqn:Dparameterization} can be preserved by considering a time-domain Wiener LPV OBF realization
\begin{subequations}
\begin{align}
\label{eqn:NparameterizationTime}
	y_{N_{K_p}}(t) &= {\textstyle \sum_{i=0}^{n_N}} \, w_i(p(t))y_\phi(t), \\
	\label{eqn:DparameterizationTime}
	y_{D_{K_p}}(t) &= {\textstyle \sum_{i=0}^{n_D}} \, v_i(p(t))y_\phi(t),
\end{align}
\end{subequations}
with $y_\phi = \Phi u$. The parameterization of $N_{K_p}$ and $D_{K_p}$ can be viewed as a bank of OBFs, whose output is weighted with parameter-dependent coefficient functions, see Figure \ref{fig:WienerOBF}.

Equations \eqref{eqn:NparameterizationTime}-\eqref{eqn:DparameterizationTime} reveal that requirements (i)-(iv) are satisfied. Requirement (v) is satisfied, w.l.o.g. by $\{ \breve{v}_{i}^{\ell} \}_{i=0}^{\ell} = \{1,0, \dots, 0 \}$. Because the set of bases is complete w.r.t. $\mc{H}_2$, hence any solution including the optimal solution of \eqref{eqn:synthesis} can be found via parameterizations \eqref{eqn:Nparameterization}--\eqref{eqn:Dparameterization} \cite{karimi2018robust}.
\begin{remark}
	Note that \eqref{eqn:NparameterizationTime} and \eqref{eqn:DparameterizationTime} depend on time-varying $p$ and characterize the global behavior of the factors $N_{K_p}$ and $D_{K_p}$. The concept in this paper is to tune the parameter-dependent coefficient functions based on their local behavior, i.e., \eqref{eqn:Nparameterization}-\eqref{eqn:Dparameterization} for constant $\mrm{p}$, in-line with the data $\mc{D}_{N, \mrm{p}_\tau}$.
\end{remark}

Algorithm \ref{alg:basisFunctionSelection} presents the selection of optimal OBFs, based on the Kolmogorov $n$-Width theory. Given a desired number of poles, an optimal set of OBFs is selected based on Fuzzy Kolmogorov c-Max (FKcM) clustering of the poles, such that the decay rate of the OBFs is minimized \cite[Chapter 8]{toth2010modeling}. 
\begin{algorithm}[]
	\caption{Basis function selection}
	\label{alg:basisFunctionSelection}
	\small
	{Choose arbitrary bases $\{ \phi_i \}_{i=0}^{n_N}$ and $\{ \varphi_i \}_{i=0}^{n_D}$, solve \eqref{eqn:synthesis} for $\theta = \{\{ \breve{w}_i^\ell \}_{\ell = 1}^m \cup \{ \breve{v}_i^\ell \}_{\ell = 1}^m \}$ and compute $\hat{N}_{K_\mrm{p}}$ and $\hat{D}_{K_\mrm{p}}$.} \\
	
	{Given $\hat{N}_{K_\mrm{p}}$ and $\hat{D}_{K_\mrm{p}}$, compute the corresponding local pole and zero variations of the controller. Choose new bases $\{\phi_i\}$ and $\{ \varphi_i \}$ based on FKcM clustering of the poles and zeros.} \\
	
	{Solve \eqref{eqn:synthesis} and compute $N_{K_\mrm{p}}$ and $D_{K_\mrm{p}}$.} \\
	
	{Stop if a desired performance and order of the bases has been achieved, otherwise go to step 2.}
\end{algorithm}

\subsection{Controller implementation}
\label{subsection:Kimplementation}
The OBF parameterizations admit a linear fractional representation (LFR). In this representation, the dependency on the scheduling variable $p$ is extracted by formulating \eqref{eqn:NparameterizationTime} and \eqref{eqn:DparameterizationTime} in terms of LTI systems, denoted $\mc{N}$ and $\mc{D}$, such that $N_{K_p} = \mc{F}_u(\mc{N}, \Delta_\mc{N}(p))$ and $D_{K_p}^{-1} = \mc{F}_u(\mc{D}^{-1}, \Delta_\mc{D}(p))$, respectively, where $\mc{F}_u$ is the upper linear fractional transformation \cite{zhou1996robust}, see Figure \ref{fig:KLFR1}. The inverse $\mc{D}^{-1}$ is obtained through partial inversion of the IO map, see, e.g., \cite[Chapter 10]{zhou1996robust}. The controller is formed through the series connection of the LFRs $\mc{N}$ and $\mc{D}^{-1}$, resulting in the LFR $\mc{K}$ such that $K_p = \mc{F}_u(\mc{K}, \mrm{diag}(\Delta_\mc{N}, \Delta_\mc{D}) )$, see Figure \ref{fig:KLFR2}.

\begin{figure}[t]
	\centering
	\begin{subfigure}[]{0.60\columnwidth}
		\centering
		\includegraphics[scale=1.00]{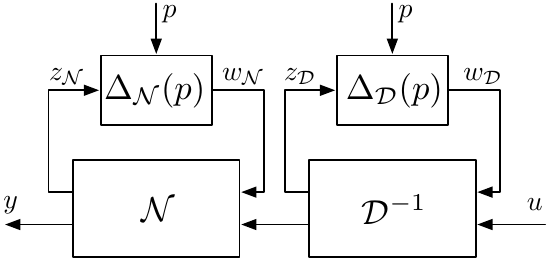}
		\caption{}
		\label{fig:KLFR1}
	\end{subfigure}
	\hfill
	\begin{subfigure}[]{0.37\columnwidth}
		\centering
		\includegraphics[scale=1.00]{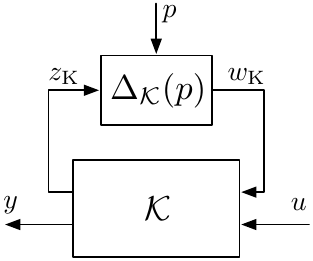}
		\caption{}
		\label{fig:KLFR2}
	\end{subfigure}
	\caption{Controller realization through (a) the series connection of the LFRs $N_{K_p} = \mc{F}_u(\mc{N}, \Delta_\mc{N}(p))$ and $D_{K_p} = \mc{F}_u(\mc{D}^{-1}, \Delta_\mc{D}(p))$ and (b) $K_p = \mc{F}_u(\mc{K}, \Delta_\mc{K})$, with $\Delta_\mc{K} = \mrm{diag}(\Delta_\mc{N},\Delta_\mc{D})$.}
	\vspace{-3ex}
\end{figure}

\section{Control design for the CMG}
\label{section:experimentalResults}
In this section, a controller is designed and implemented on the CMG. Although the theory in this paper is presented in continuous-time, with this example we show that a discrete-time application is possible. The system identification and controller design are performed at a sampling rate of $200$ Hz.

\subsection{Frequency-domain measurements}
\label{subsection:identification}
As described in Section \ref{section:ProblemFormulation}, the dynamics of the CMG are dependent on the velocity of the disk. It is therefore natural to consider the velocity $\dot{q}_1(t) = p(t)$ as a scheduling variable. The disk velocity operates in the range $\mbb{P} = [30, 50]$ rad/s. To identify the local behavior at different disk velocities, an equidistant grid $\mc{P} = \{30, 40, 50 \}$ is chosen. As the gyroscope is inherently an unstable system, the measurements are performed in closed-loop using a stabilizing LTI controller. 

The coprime factors $N_{G_\mrm{p}}(i\omega)$ and $D_{G_\mrm{p}}(i\omega)$ can be calculated from the estimates of the process sensitivity $S_\mrm{p}G_\mrm{p}$ and sensitivity $S_\mrm{p}$, respectively \cite{karimi2018robust}. This is achieved by estimating the fFRF of the mappings $d \mapsto y$ and $d \mapsto u_G$, respectively, in Figure \ref{fig:internal_stability}. During a closed-loop experiment, the system is excited by a white-noise disturbance signal $d$. The position of gimbal $A$ is measured with an optical encoder. Data records with a length of $240000$ samples are collected for each operating point $\mrm{p} \in \mc{P}$. 

The obtained fFRFs are estimated using the empirical transfer function estimate, using a Hanning window, and contain $1000$ frequency points per operating point. Figure \ref{fig:Gbode} shows the estimated fFRFs $G_\mrm{p}$. The figure highlights that the system is subject to a relatively high noise level, which has a significant effect at higher frequencies. The scheduling dependency is also clear to see, which manifests in terms of a shift in the resonance frequencies and the low-frequency gain.
\begin{figure}[]
	\centering
	\includegraphics[scale=1.00]{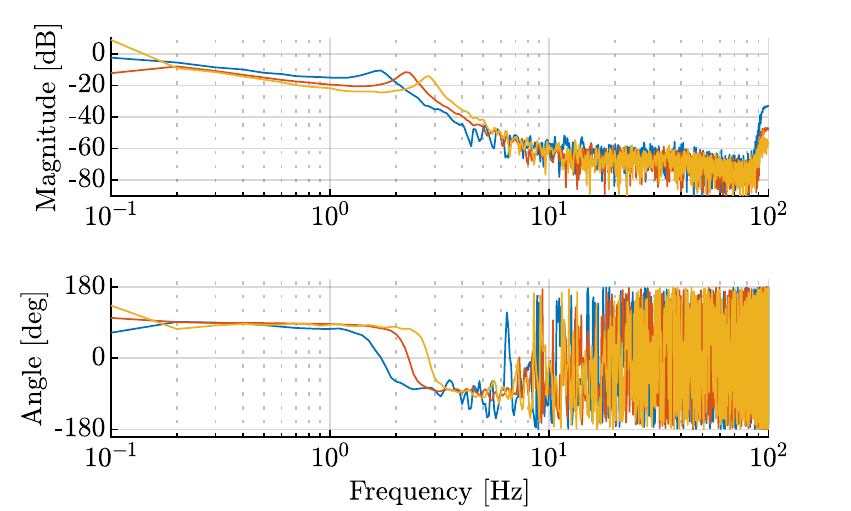}
	\caption{Bode plot of the estimated fFRFs of $G_\mrm{p}$ for the three grid points $\mrm{p} \in \mc{P} = \{30,40,50\}$ in blue, orange and yellow, respectively. A parameter-varying shift in resonance frequencies and low-frequency gain is observed.}
	\label{fig:Gbode}
	\vspace{-3ex}
\end{figure}

\subsection{Data-driven controller synthesis}
\label{subsection:controllerImplementation}
The goal is to control the position $q_4$ of gimbal $A$ by actuating gimbal $C$ through torque $\tau_2$. To highlight the parameter dependence, the objective is to track a reference signal subject to variations in the disk velocity. To specify this objective in terms of control design, consider the full 4-block shaping problem in Figure \ref{fig:internal_stability}. Based on the fFRFs in Figure \ref{fig:Gbode}, the first resonance occurs at $1.7$ Hz. The shaping filters are designed such that a bandwidth of $0.75$ Hz is achieved. The sensitivity is shaped to provide a lower bound on the bandwidth and to limit the overshoot by providing an upper bound of $6$ dB for higher frequencies. Integral action is desired to achieve 0 steady-state error. To suppress the effects of measurement noise while also limiting high-frequent control actions, a high-frequent roll-off is enforced into the controller by shaping the control and complementary sensitivities. Shaping the complementary sensitivity also provides an upper bound on the achieved  bandwidth. The process sensitivity is restricted to lie below $0$ dB to limit the amplification of disturbances. 

Using the approach presented in this paper, an LPV and LTI controller are synthesized, for which the results are given in Figures \ref{fig:localResults} and \ref{fig:Kbode}. Both controllers are parameterized by discrete-time Laguerre bases of orders $n_K = n_D = 5$ with pole $z = 0.7$. The LPV controller has affine scheduling-dependance, and the LTI controller is scheduling-independent. The achieved performance levels are $\gamma_\mathrm{LPV} = 1.2097$ and $\gamma_\mathrm{LTI} = 3.1792$. The LTI controller does not meet the performance criteria for all operating points and, therefore has to sacrifice performance in order to achieve robust performance. The LPV controller achieves good performance for the considered operating space by compensating for the parameter-dependent low-frequency gain and resonance behavior.
\begin{figure}[]
	\centering
	\includegraphics[scale=1.00]{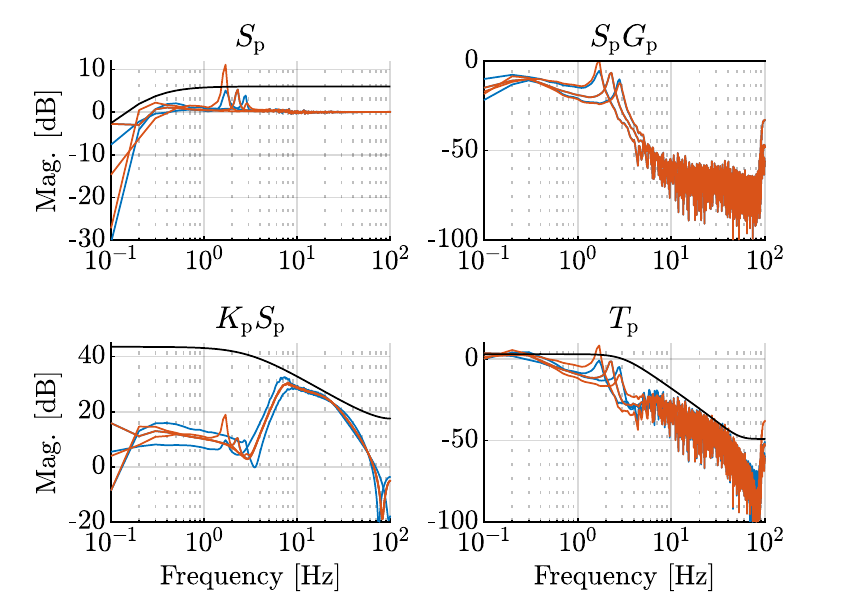}
	\caption{Magnitude plots of the fFRFs of the 4-block \eqref{eqn:four-block}. The LPV and LTI designs are shown in blue and orange, respectively. The weighting filters are shown in black. The LTI controller design does not meet the performance specifications.}
	\label{fig:localResults}
	\vspace{-3ex}
\end{figure}
\begin{figure}[]
	\centering
	\includegraphics[scale=1.00]{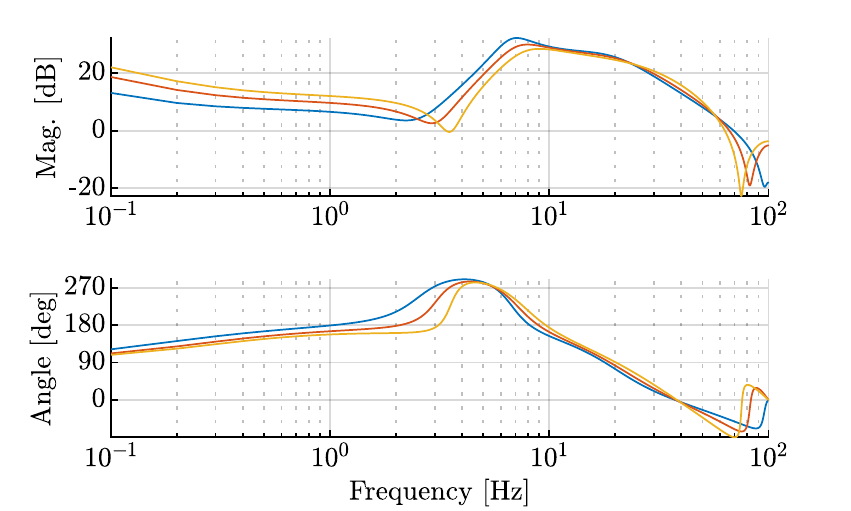}
	\caption{Magnitude and phase plots of the LPV controller for frozen scheduling-values $\mc{P} = \{30, 40, 50\}$ in blue orange and yellow, respectively. The LPV controller compensates the parameter-varying low-frequency gain and resonance behavior observed in Figure \ref{fig:Gbode}.}
	\label{fig:Kbode}
	\vspace{-3ex}
\end{figure}

\subsection{Results}
\label{subsection:experimentalResults}
First, the tracking performance is evaluated locally, when the scheduling variable operates at constant velocities $\mc{P} = \{30, 40, 50 \}$ rad/s. Figure \ref{fig:expLocalRef} shows the measured step responses using the designed LPV and LTI controllers. The main differences are observed for $\mrm{p}=30$ and $\mrm{p}=50$ rad/s. At $\mrm{p}=30$ rad/s, the step response shows a significant oscillation when using the LTI controller. This oscillation corresponds to the resonance frequency at $1.7$ Hz in Figure \ref{fig:localResults} and it is significantly larger compared to the LPV case. For $\mrm{p}=50$ rad/s, a slightly higher bandwidth is achieved when using the LPV controller, which corresponds to a faster rise and settling time. Finally, the responses when using the LPV controller are very consistent, with only a small variation in settling time.
\begin{figure}[]
	\centering
	\includegraphics[scale=1.00]{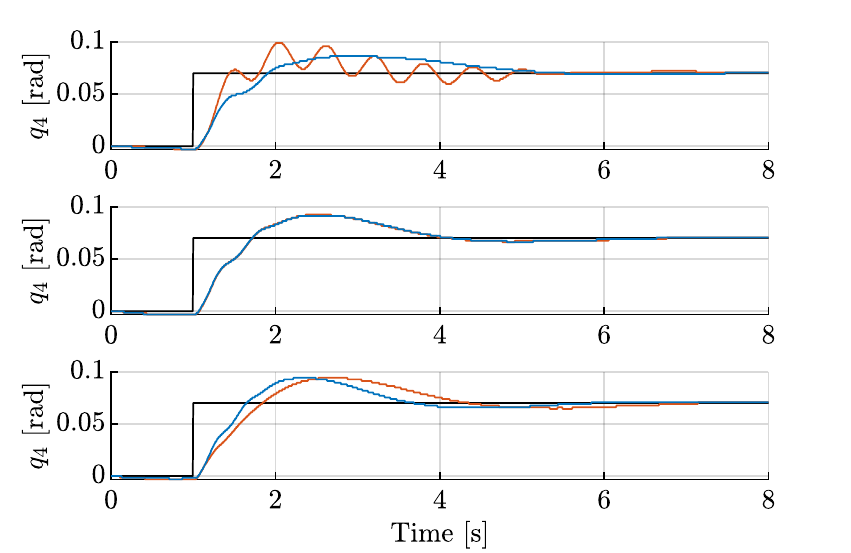}
	\caption{Measured local step responses of the CMG for constant scheduling variables $\mc{P} = \{30, 40, 50\}$ rad/s, from top to bottom, respectively. The angle of gimbal $A$ is shown in blue and orange when using the LPV and LTI controller, respectively. The LTI design loses performance for $\mrm{p} \in \{30, 50\}$ rad/s, whereas the LPV design displays consistent results for the considered operating points.}
	\label{fig:expLocalRef}
	\vspace{-3ex}
\end{figure}

Next, the performance is evaluated for a time-varying scheduling variable. A Square wave reference signals, filtered with a 3rd order low-pass filter with a cut-off frequency of $0.7$ Hz, are used to challenge the system. The amplitude of the reference is $15^\circ$. The scheduling variable, i.e. the disk velocity, tracks a similar, but faster square wave trajectory in the range $\mbb{P} = [30, 50]$ rad/s. Implementation of the controller is done according to the LFR representation described in Section \ref{subsection:Kimplementation}, where the controller is scheduled at each sampling interval.

Figure \ref{fig:expResults} shows the reference signal, tracking performance, scheduling variation and control effort for the designed LPV and LTI controllers. The results indicate that the LPV controller performs significantly better than the LTI controller. A reduction in overshoot and settling time are observed. More specifically, we obtain a $39 \%$ and $33 \%$ decrease between the $\ell_2$ and $\ell_\infty$ norms of the error signals, respectively. These results experimentally validate the capabilities of the proposed control methodology, including the benefit of using an LPV controller over an LTI controller for the CMG. However, it is imperative to note that stability and performance guarantees are provided only locally. Hence, stability and performance of the nonlinear system can only be guaranteed for sufficiently slow variations of the scheduling variable.
\begin{figure}[t]
	\centering
	\includegraphics[scale=1.00]{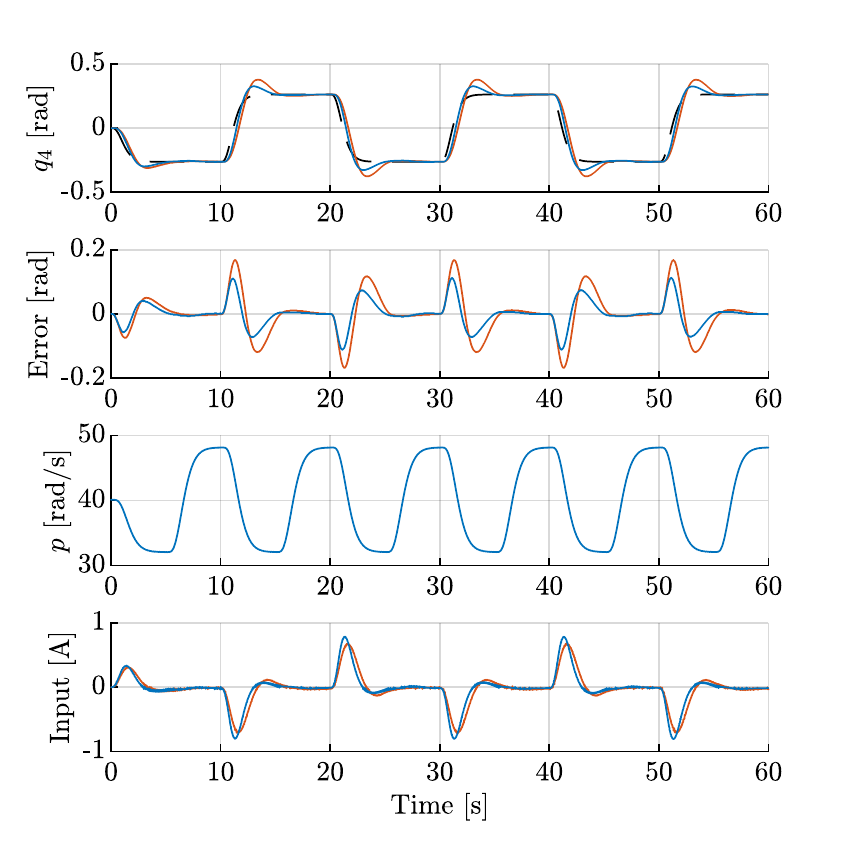}
	\caption{Experimental results of the CMG. The top figure shows the reference (black) and the angle of gimbal $A$ when using the LPV and LTI controllers in blue and orange, respectively. The other figures shows the error, scheduling and input signals. The LPV design significantly improves the performance by decreasing the overshoot.}
	\label{fig:expResults}
	\vspace{-3ex}
\end{figure}

\section{Conclusion}
\label{section:conclusions}
The LPV controller synthesis approach in this paper enables the design of operating condition-dependent controllers directly from frequency-domain data. Experimental demonstrations on a control moment gyroscope show that significant increase in performance can be achieved via the proposed approach for operating condition-dependent systems. In comparison to existing methods in the literature, this approach enables the design of rational LPV controllers, for which local stability and performance analysis certificates are provided. Future research aims at global stability and performance guarantees.


%

%
%
%
%
%




\bibliographystyle{IEEEtran}
\bibliography{IEEEabrv,TCST2020}{}

\end{document}